\documentclass[letterpaper,twocolumn,10pt]{article}
\usepackage{usenix}

\usepackage{amsmath}
\usepackage{graphicx}
\usepackage{latexsym}
\usepackage{bm}
\usepackage{enumitem} 
\usepackage{amsfonts}
\usepackage{listings} 
\usepackage{tikz}
\usepackage{comment}
\usepackage{caption}
\usepackage{array}
\usepackage{xspace}
\usepackage{psfrag}
\usepackage{subfigure}
\usepackage{color}
\usepackage{algorithmicx}
\usepackage{algpseudocode}
\usepackage{alltt}
\usepackage{multirow}
\usepackage{array}
\usepackage{wasysym}
\usepackage{rotating}
\usepackage{datetime}
\usepackage{tabularx}
\usepackage{mathrsfs}
\usepackage{mdwtab}
\usepackage{pifont}
\usepackage{amsthm}
\usepackage{epsfig}
\usepackage{endnotes}
\usepackage{booktabs}
\usepackage{makecell}
\usepackage{graphicx}
\usepackage{epstopdf}
\usepackage[para,online,flushleft]{threeparttable}
\usepackage{tikz}

\usepackage{algorithm}
\usepackage{algpseudocode}

\usepackage{url}
\newcolumntype{Y}{>{\centering\arraybackslash}X}

\newtheorem{theorem}{Theorem}
\newtheorem{lemma}[theorem]{Lemma}

\newcommand{\myname}[1]{\texttt{AnyPro}\xspace}

\newcommand{\minyuan}[1]{#1}

\newcommand{\yuning}[1]{\textcolor{red}{[Yuning: #1]}}

\newcommand{\ie}{\textit{i.e.}\xspace}

\newcommand{\eg}{\textit{e.g.}\xspace}

\newcommand{\etal}{\textit{et al.}\xspace}

\newcommand{\fp}{\vspace*{0.05in}\noindent}
\newcommand{\newparagraph}[1]{\fp {\bf #1}~}
\newcommand{\lar}{normalized objective\xspace}
\newcommand{\Lar}{Normalized objective\xspace}
\newcommand{\LAR}{Normalized Objective\xspace}
\newcommand{\sys}{\texttt{AnyPro}\xspace}
\newcommand{\pt}{ingress\xspace}
\newcommand{\scan}{\textit{max-min polling}\xspace}

\newlength{\onecolgrid}
\newlength{\twocolgrid}
\newlength{\threecolgrid}
\newlength{\fourcolgrid}
\colorlet{green2}{blue!15!green!85!}

\begin{document}


\title{AnyPro: Preference-Preserving Anycast Optimization \\ based on Strategic AS-Path Prepending}


\author{
{\rm 
Minyuan Zhou\textsuperscript{1,4}\textbf{*}$^{\dagger}$, 
Yuning Chen\textsuperscript{2,4}\textbf{*}$^{\dagger}$, 
Jiaqi Zheng\textsuperscript{1}, 
Yifei Xu\textsuperscript{3,4}$^{\dagger}$, 
Pan Hu\textsuperscript{4}, 
Yongping Tang\textsuperscript{4}, 
Wendong Yin\textsuperscript{4}, 
}\\
{\rm 
Jie Lin\textsuperscript{4}, 
Qingyan Yu\textsuperscript{4}, 
Yuanchao Su\textsuperscript{4}, 
Guihai Chen\textsuperscript{1}, 
Wanchun Dou\textsuperscript{1}, 
Songwu Lu\textsuperscript{3}, 
Wan Du\textsuperscript{2}
}\\
\textsuperscript{1}State Key Laboratory for Novel Software Technology, Nanjing University\\
\textsuperscript{2}University of California, Merced\\
\textsuperscript{3}University of California, Los Angeles\\
\textsuperscript{4}Alibaba Cloud\\
} 

\maketitle
\thispagestyle{empty}

\begingroup
\renewcommand\thefootnote{}
\footnotetext{\textbf{*}The first two authors contributed equally to the work.}
\footnotetext{$^{\dagger}$Work done during the internships of Yuning Chen, Minyuan Zhou, and Yifei Xu at Alibaba Cloud.}
\endgroup

\begin{abstract}
Operating large-scale anycast networks is challenging because client-to-site mappings often misalign with operator's expectation due to opaque inter-domain routing.
We present \sys, the first system to unlock the full potential of AS-path prepending (ASPP), efficiently deriving globally optimal configurations to steer clients toward performance-optimal sites at scale.
\sys first employs an efficient polling mechanism to identify all clients sensitive to ASPP.
By analyzing the routing changes during the process, the system derives a set of ASPP constraints that guide client traffic toward the desired sites.
We then formulate the anycast optimization problem as a constraint-based program and compute optimal ASPP configurations. Extensive evaluation on a global testbed with 20 PoPs demonstrates the effectiveness of \sys: it reduces the 90th percentile latency by 37.7\% compared to baseline configurations without ASPP. Furthermore, we show that \sys can be integrated with PoP-level anycast optimization techniques to achieve additional performance gains. 

\end{abstract}

\section{Introduction}

IP anycast is a foundational and widely adopted technique in today's Internet infrastructure, which involves advertising the same IP prefix from multiple geographically distributed Points of Presence (PoPs). This design enhances key service qualities including resilience to failures, distributed load balancing, and reduced latency by routing users to topologically or geographically proximate instances~\cite{li2023rovista,koch2023painter,koch2021anycast,zhang2021anyopt,Zhou:sigcomm2023,fastroute,calder2015analyzing,Stephen2019Taming,ballani2005towards}. Client-to-PoP routing is not governed by explicit centralized rules, but rather by the policy-driven inter-domain routing behavior of BGP, which selects paths based on configured policies and AS-level relationships. Due to its inherent flexibility, scalability, and robustness~\cite{zhu2022best}, IP anycast has become a critical technology underlying major global services, including root DNS systems~\cite{dnstwoday-liu-pam07,akamaiDNS-schomp-sigcomm20,Sarat2006On}, large-scale DDoS mitigation platforms~\cite{anycastddos-moura-imc16}, and content delivery networks (CDNs)~\cite{de2020global}, where improving performance and reliability is essential.

In a typical anycast deployment, each site serves a group of client IPs -- forming what is known as its catchment. Ideally, these catchments align with geographic proximity, topological closeness, or minimum-delay partitions to ensure optimal quality of service, especially in terms of latency. However, anycast does not always achieve such performance. Due to the nature of BGP routing, clients may be directed to geographically distant sites~\cite{ballani2005towards,ballani2006a,Sarat2006On,dnstwoday-liu-pam07,Li2018Internet}, resulting in significant path inflation that can add over 100 ms of latency and severely degrade user experience.

Network operators often predefine optimal client-to-site mapping (catchment) based on operational experience or application needs~\cite{zhang2021anyopt}. 
Achieving optimal mapping has always been the north star for anycast operations, but existing approaches are still far from it due to limited visibility into cross-AS topology and external routing policies. 
For instance, embedding the PoP’s geographic location in a BGP announcement is an intuitive idea~\cite{Li2018Internet}, but it demands BGP attribute extensions across multiple routing domains, which is operationally impractical. 
AnyOpt~\cite{zhang2021anyopt}, a recent optimization framework, selects a subset of sites through pairwise preference discovery. While it improves global catchment efficiency, the cost is sacrificing the latency of a significant fraction of clients, a compromise operators are unwilling to make.

In this landscape, AS-path prepending (ASPP) stands out as a remarkably versatile mechanism for fine-grained routing control~\cite{marcos2020path}.
Unlike more intrusive methods, ASPP extends AS-path lengths at selected PoPs to subtly steer  BGP route selection for certain ASes. 
However, this powerful tool has never been systematically used for large-scale anycast tuning. 
The key challenge lies in the huge search space. Unlike binary site selection, ASPP introduces exponentially more combinations, making pairwise scan~\cite{zhang2021anyopt} computationally infeasible.
Consequently, ASPP remains an underutilized yet highly promising mechanism for anycast tuning.

In this paper, we propose \sys, which can achieve globally optimal ASPP configurations in a practically efficient way.
\sys operates through three core phases: \textit{max-min polling} identifies ASPP-sensitive clients and establishes preliminary constraints; \textit{optimization solving} maximizes constraint satisfaction; and \textit{contradiction resolution} uses binary search to refine conflicting constraints. The final output is an optimal set of ASPP configurations per PoP-transit access point that yield the intended client-to-ingress mappings.


Overall, \sys provides a comprehensive framework for anycast optimization with three key capabilities: 
(1) systematic detection of ASPP-sensitive clients, 
(2) efficient identification of configuration constraint to satisfy routing preferences,
and (3) derivation of the optimal ASPP configurations. 
We evaluate \sys on a global testbed of 20 production PoPs. 
To further evaluate \sys's versatility, we integrate it with AnyOpt -- a state-of-the-art anycast optimization technique that operates through strategic PoP enablement. This combined approach yields remarkable performance gains: the 90th percentile latency is reduced to 58.0 ms, representing a significant improvement over baseline measurements. Regional-scale validation in Southeast Asia demonstrates particularly impressive results, with a 17.9\% improvement in the matching accuracy. These multi-scale experiments confirm \sys's effectiveness across different network topologies and operational requirements, while its constraint-based optimization framework provides network operators with both flexibility and precise control over anycast routing behavior.

The contribution can be summarized as follows.
\sys pioneers production-grade anycast catchment optimization by introducing a robust and deterministic framework that automatically derives optimal AS-path prepending configurations aligned with operator performance objectives. It establishes comprehensive theoretical foundations and a practical methodology, supported by an efficient algorithm capable of identifying ASPP-sensitive clients, generating preference-aware constraints --- even under third-party or middle-ISP changes --- and computing near-optimal configurations through iterative refinement. Evaluated on production networks, the system demonstrates substantial gains in catchment efficiency and RTT across diverse operational scenarios, while supporting peering-aware configurations and enabling future integration of advanced BGP controls.


\section{Background and Motivation}\label{sec:motivation}


Global network services require providers to deploy distributed servers, ensuring clients connect to optimal sites for ideal Quality of Service (QoS)~\cite{xu2024cloudeval, xu2023cloudeval}. Anycast simplifies this by advertising the same IP prefixes across all sites, leveraging BGP to route clients to the nearest or best-performing Point of Presence (PoP). We define a PoP as the physical access point and its connected clients as the catchment of that PoP. Each PoP may have multiple transit providers, and we refer to an ingress as a unique (PoP, transit provider) pair, representing a distinct entry point into the network.

Prior work has optimized anycast performance through various approaches, such as deploying additional PoPs~\cite{de2017Anycast}, expanding peering relationships with other ASes~\cite{calder2015analyzing, koch2021anycast}, and encoding PoP locations in BGP communities~\cite{Li2018Internet}. However, these methods often lack deployability in practice. AnyOpt~\cite{zhang2021anyopt} introduced PoP-level optimization by selectively enabling or disabling sites, but its reliance on pairwise BGP experiments creates operational overhead that limits scalability. To address this, we propose a finer-grained control mechanism: instead of operating at the PoP level, we optimize announcements at the ingress level --- the granularity of a (PoP, transit provider) pair --- enabling more flexible and scalable anycast routing.

\begin{figure*}[htb]
\centering
\includegraphics[width=1\linewidth]{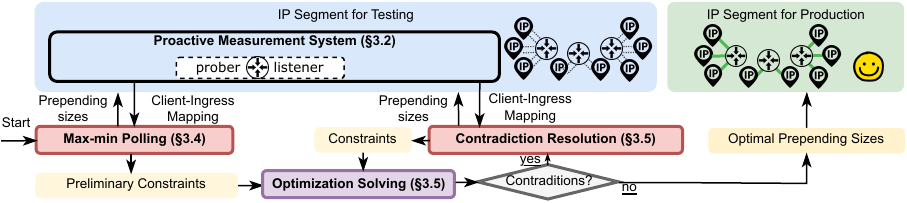}
\vspace{-8pt}
    \caption{\sys system overview.
      \label{fig:overview}
    }
\vspace{-11pt}
\end{figure*}

Typically, network operators employ two mechanisms to customize route announcements.
{\bf (1) Community}~\cite{Tony1997Community}. BGP communities act as standardized tags that convey routing instructions to neighboring ASes, enabling functions like route propagation control and path selection optimization~\cite{krenc2023coarse, IMC2021-ASComm}. However, their implementation lacks universal consistency --- different networks assign distinct community numbers for identical functions~\cite{Telia2024BGP,Cogent2024BGP}, creating significant operational complexity. This variability makes community-based routing configurations particularly challenging to manage at scale. Due to these inherent complexities, we exclude BGP community management from our primary focus and address it separately in \S\ref{sec:discuss community}.
{\bf (2) AS-path prepending}~\cite{Gao2005Interdomain}. AS-path prepending (ASPP) is a BGP traffic engineering mechanism in which an AS artificially extends its AS-path length by inserting duplicate instances of its AS number. Service providers strategically vary prepending lengths when announcing anycast prefixes to transit providers (\pt), leveraging BGP's path selection algorithm, which prioritizes routes with shorter AS-paths. This enables fine-grained control over inbound traffic flows~\cite{li2023rovista}. Importantly, ASPP policies exhibit remarkable stability across the Internet --- only 0.3\% of paths show prepending changes, and most modifications are temporary, reverting within days~\cite{zhu2022best}. This stability ensures that our optimization operates in a predictable environment, free from interference by sporadic prepending adjustments in other networks (as we detailed in \S\ref{sec:anypro}).

While ASPP offers significant potential for catchment optimization, its effective implementation presents three key challenges: (1) managing computational complexity, (2) inferring sufficient routing topology context to ensure accurate client-to-ingress mappings, and (3) constraining the total number of required configuration experiments. To address these challenges, we require a systematic approach that can both efficiently analyze ASPP's impact on individual client routing and determine globally optimal prepending configurations across all network ingresses.

\section{AnyPro} \label{sec:anypro}

This section presents an overview of \sys, its practical challenges, methods for discovering ASPP-sensitive clients, optimization for preference-preserving configurations, and the rationale behind its operational algorithms.

\subsection{Overview}
\label{sec:sysover}

\textbf{Real-world anycast testbed.}
Our anycast testbed operates on a production-grade network comprising 20 globally distributed PoPs, each connected to 1\textasciitilde3 transit providers (totaling 38 ingresses, as listed in Appendix~\ref{appendix:infra}). These PoPs peer with other ASes through IXPs, ensuring realistic routing conditions. The infrastructure supports two IP segments: one for live traffic and one dedicated to experiments. 
While each segment may employ independent ASPP configurations and exhibit distinct client-to-ingress mappings, identical settings always yield reproducible mappings since they share the same backbone network. 
Using the test segment, we deploy a proactive measurement system that efficiently characterizes global catchments for arbitrary ASPP configurations without exhaustive testing, enabling scalable optimization.

\noindent\textbf{Max-min polling identifies ASPP-sensitive clients while generating preference-preserving routing constraints.}
To derive sufficient conditions for ingress-level preference compliance, we first design an exploration function, \scan, which systematically probes the configuration space by: (1) initializing all PoP-transit prepending lengths to their maximum values, then (2) iteratively reducing to the minimum length. This scanning process identifies all feasible ingress preferences per client IP (with formal proof in \S\ref{sec:model}) while simultaneously generating preliminary preference-preserving constraints.



\noindent\textbf{Constraint contradiction resolution.} 
While \scan generates preliminary preference-preserving constraints for individual clients, combining constraints across all clients reveals frequent contradictions. This occurs because the initial exploration only tests prepending lengths at the boundaries (0 and $MAX$), resulting in coarse-grained length differences (limited to 0 and $MAX$)
that fail to capture intermediate optimal values. To address this, we develop an efficient contradiction resolution algorithm that employs binary search to iteratively refine the prepending length difference thresholds. This approach produces a tightly constrained optimization problem while maintaining computational efficiency --- reducing the search space by orders of magnitude compared to brute-force methods.

\begin{figure}[tb]
    \centering
    \includegraphics[width=\columnwidth]{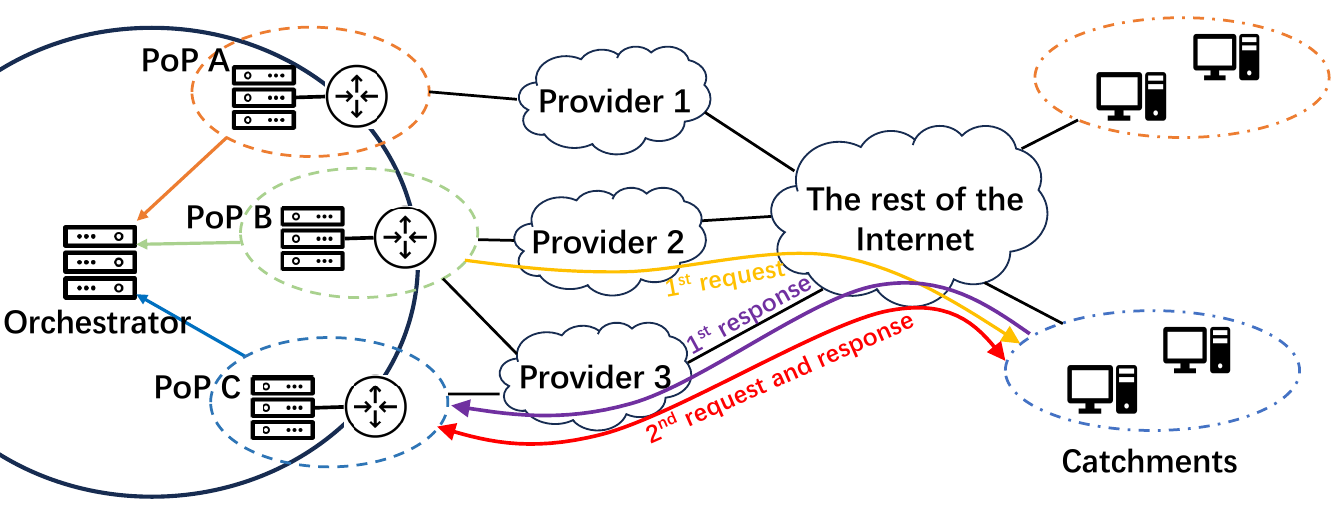}
    \vspace{-2pt}
    \caption{Illustration of the real-world testbed we used for anycast measurements.}
    \vspace{-2pt}
    \label{fig:measure_sys}
\end{figure}

\subsection{Proactive measurement system} \label{sec:mearsys}
Although traditional passive measurement with client request logs can reflect real-time ingress behavior statistics, it has two major limitations.
First, client request logs in some client IPs can be sparse, especially considering the peak and off‐peak hours within different time zones. Therefore, it requires significant time to output a complete client-ingress mapping for all clients. 
Second, client request logs are collected directly from production networks, which prevents trying risky configurations and collecting the feedback effect on real networks.

To ensure the generalizability of our experimental findings across diverse network environments, we construct a representative and stable IP set derived from the authoritative ISI IPv4 hitlist~\cite{Fan10a} --- the most comprehensive public source of active IPv4 addresses, offering broad coverage across geographic and topological regions. This IP selection strategy, combined with a client-agnostic proactive measurement approach that operates without relying on client requests, enhances the universal applicability of our experimental results~\cite{Verfploeter2017de,sommese2020manycast2}. 
To further improve reliability, we filter the initial IPs through week-long active probing, retaining only addresses with under 10\% packet loss, which excludes unstable nodes and strengthens measurement validity. 
The final set comprises \textasciitilde2.4M IP addresses, a substantial portion of which belong to network middleboxes. 


To measure client-ingress mapping and client-to-PoP RTTs, we developed a custom measurement tool (Figure~\ref{fig:measure_sys}). 
Each ingress hosts a prober-listener pair that exchanges ICMP packets in an anycast format. The measurement process begins with proactive probing: anycast PoPs send ICMP requests (yellow lines) with anycast source addresses to target clients. Responses (purple lines) reveal catchment mappings by routing to the nearest enabled PoP. Upon receiving a response, the ingress immediately issues a follow-up ICMP request containing a unique identifier and timestamp. When the response (red line) returns, we compute RTT from the timestamp delta and log the result alongside ingress metadata (transit providers, peers, IXPs) to inform announcement optimization. This dual-phase approach simultaneously captures catchment boundaries and latency while minimizing measurement overhead.

\subsection{Practical challenges}\label{sec:pratical_challenges}
We have introduced the overview of \sys. However, to make it practical, we must address the following challenges.

\newparagraph{Determining client-to-ingress mappings across all possible prepending configurations presents significant computational and methodological challenges.} 
In anycast networks, client-to-ingress mappings are determined by BGP routing decisions, which respond dynamically to AS-path prepending adjustments. For a network with $n$ ingresses (each supporting $m$ prepending lengths), exhaustive testing of all possible configurations would require $\mathcal{O}(m^{n})$ BGP experiments. Given that each experiment incurs nontrivial convergence delays (often minutes), this brute-force approach becomes computationally intractable at scale, necessitating a more efficient optimization methodology.

\noindent{\bf Constraint contradiction resolution incurs significant computational overhead.} 
When consolidating preliminary preference-preserving constraints across all clients, inherent contradictions frequently arise due to competing routing objectives. Resolving each contradiction through naive enumeration would necessitate $\mathcal{O}(m)$ BGP experiments per conflict (systematically testing each ingress's prepending length from MIN to MAX). This approach becomes operationally infeasible for three key reasons: (1) BGP convergence delays introduce minutes of network churn per experiment, (2) the quadratic scaling of total experiments ($\mathcal{O}(m \cdot k)$ for $k$ contradictions) becomes prohibitive in large networks, and (3) transient Internet routing fluctuations may invalidate results before completion. Our measurements show that even mid-sized deployments (e.g., 20 PoPs with 38 ingresses) could require weeks of continuous testing to resolve all conflicts through brute-force methods --- an untenable proposition for production networks requiring real-time optimization.

\subsection{Max-min Polling}\label{sec:model}
A client IP is considered ASPP‑sensitive only when it can reach the prefix through at least two distinct ingresses.
Among all ingresses, the client IPs prefer one over others due to service quality discrepancies. The matrix that records whether each client prefers an ingress or not is the desired client-ingress mapping, denoted as $\boldsymbol{M^{*}}$.
To find out all ASPP-sensitive clients and their potential routes, 
we develop an efficient \scan method as depicted in Algorithm~\ref{algo:1}, which produces two outcomes:

\noindent\textbf{Outcome 1: Discovering ASPP-sensitive clients.}
\textit{Max-min polling} is executed by initially setting the prepending length to \textit{MAX} for all ingresses, and iteratively adjusting each to zero while others remain \textit{MAX}, where \textit{MAX} is an upper bound of the prepending length, while \textit{MIN} is set as zero.
Initially, all prepending lengths are set to \textit{MAX} (line 1).
Under this initial configuration, the client-ingress mapping, denoted as $\boldsymbol{M}$, serves as a baseline mapping (line 2).
We iteratively tune the ASPP configuration of each ingress and observe the corresponding client-ingress mapping changes (line 3–8).
For each ingress, we first eliminate the prepending (line 4).
At this moment, since the prepending lengths of others are kept at $\textit{MAX}$, this ingress becomes the most preferred, in terms of AS-path length, and then we measure the client-ingress mapping, denoted as $\boldsymbol{M’}$ (line 5).
Comparing $\boldsymbol{M’}$ and $\boldsymbol{M}$, we identify the clients whose ingresses change.
These clients are identified as ASPP-sensitive, as their routing behavior responds directly to alterations in the prepending configuration of the specific ingress.

\vspace{0.1in}
\begin{algorithm}
\caption{\scan execution and processing}
\begin{algorithmic}[1]
\algtext*{EndFor}
\Statex \hspace{-\algorithmicindent}\textbf{Input:} All ingresses $\boldsymbol{P}= \{ p_{i,j} \}$ for the $i_{th}$ PoP and $j_{th}$ transit; the desired client-ingress mapping $\boldsymbol{M^{*}}$
\Statex \hspace{-\algorithmicindent}\textbf{Output:} Pairwise preference constraints $\boldsymbol{\Gamma}$ for $\boldsymbol{M^{*}}$
\State Set the prepending length of each $p_{i,j} \in \boldsymbol{P}$ to $MAX$.
\State Gather client-ingress mapping $\boldsymbol{M}$ with \S3.2. 
\For{each $p_{i,j} \in \boldsymbol{P}$ }
    \State Set the prepending length of $p_{i,j}$ to 0.
    \State Gather client-ingress mapping $\boldsymbol{M'}$ with \S3.2.
    \State Compare $\boldsymbol{M'}$ and $\boldsymbol{M}$, generate pairwise preference constraints $\gamma$ to achieve the desired mapping $\boldsymbol{M^{*}}$.
    \State $\boldsymbol{\Gamma} = \boldsymbol{\Gamma} \cup \gamma$.
    \State Set the prepending length of $p_{i,j}$ to $MAX$.
\EndFor
\end{algorithmic} \label{algo:1}
\end{algorithm} 
\vspace{0.1in}

\begin{figure}[tb]
    \includegraphics[width=.98\linewidth]{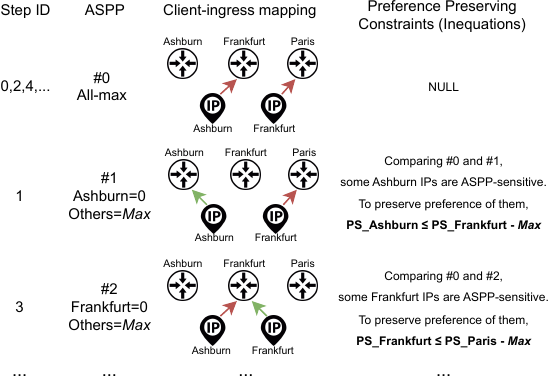}
    \vspace{4pt}
    \caption{An illustrative example of \scan.
      \label{fig:poll_example}
    }
\end{figure}

To prove the completeness of \scan in identifying ASPP-sensitive clients, we now introduce a lemma about \scan in a pair of ingresses.
\begin{lemma}[Completeness of \scan for a pair of ingresses] \label{lem:1}
For any two ingresses $p_{i,j}$ and $p_{m,n}$, \scan can effectively explore all potential routes for all clients by shrinking the length of the decision space
from an interval $(s_{i,j},s_{m,n}\in [MIN, MAX])$ to two points $(s_{i,j},s_{m,n} \in \{MIN, MAX\})$,
where $s_{i,j}$ and $s_{m,n}$ is the prepending length of $p_{i,j}$ and $p_{m,n}$, respectively.
\end{lemma}
\begin{proof}
    During \scan, the prepending length difference between any two ingresses can be defined as $s_{i,j}-s_{m,n} = \{\textit{MIN}-\textit{MAX}$, 0, $\textit{MAX}-\textit{MIN}\}$.
    Therefore, considering that the preference of the route selection is monotonic to the prepending length difference, all possible routes are explored by \scan.
\end{proof}
Based on the discussion of two ingresses shown in Lemma~\ref{lem:1}, we can further derive a theorem to extend Lemma~\ref{lem:1} to a general case with multiple ingresses.

\begin{theorem}
    \scan can explore all ASPP-sensitive clients as well as their potential routes, where the prepending length of each \pt belongs to the interval [\textit{MIN}, \textit{MAX}].
\end{theorem}
\begin{proof}
    The set of all ingresses can be divided into multiple pairs. Each pair can be explored independently by \scan as discussed in Lemma \ref{lem:1}. By applying \scan to each pair, we ensure that the ASPP-sensitive clients and their potential routes are thoroughly explored. 
\end{proof}

\noindent\textbf{Outcome 2: Revealing preliminary preference-preserving constraints.} 
During \scan, apart from discovering all sensitive clients, we also obtain all \textit{candidate ingresses} for each client.
Among them, we generate a set of inequation conditions to drive each client into the desired ingress.
These inequation conditions are defined as \textit{preliminary preference-preserving constraints}.
For each client, pairwise comparisons between candidate ingresses can establish a critical threshold \(\Delta s^*\) that governs routing preference: if the operator intends to prioritize ingress A over ingress B, constraints are formulated to ensure A's ASPP configuration maintains a prepending advantage of at least \(\Delta s^*\) than B.
However, since \scan tests only two ASPP configurations per ingress (\ie, 0 and MAX prepending), the precise \(\Delta s^*\) cannot be empirically determined, resulting in preliminary constraints that approximate routing behavior. 

Figure~\ref{fig:poll_example} illustrates an example of deriving preliminary preference-preserving constraints during \scan process (lines 6–7, Algorithm~\ref{algo:1}).
Initially, each ingress is configured to prepend $MAX$ length.
When a client $c_k$ in Ashburn switches from ingress Frankfurt ($p_{Frankfurt}$) to Ashburn ($p_{Ashburn}$) (for simplicity, we do not specify the transit provider here) after removing prepending at Ashburn ($s_{Ashburn}$) while others remain $MAX$, it reveals a partial preference order: $p_{Ashburn} \succ_{c_k} p_{Frankfurt}$ under $s_{Ashburn}=0$, $s_{Frankfurt}=MAX$, but $p_{Frankfurt} \succ_{c_k} p_{Ashburn}$ under uniform MAX prepending.
To enforce $c_k$'s desired access to $p_{Ashburn}$, the preference-preserving constraint $s_{Ashburn} \leq s_{Frankfurt} - MAX$ is derived. 
This ensures $p_{Ashburn}$'s dominance by maintaining a prepending gap exceeding $\Delta s^*$ between $s_{Ashburn}$ and $s_{Frankfurt}$, while intermediate configurations ($s_{Ashburn} = \frac{MAX}{2}, s_{Frankfurt}=MAX$) are excluded due to nondeterministic routing behavior (We also show why $min-max$ polling fails to satisfy our requirements in Appendix~\ref{appendix:max-min}). 

We now establish a formal proof demonstrating that for any pair of candidate ingresses, there exists a unique critical threshold \(\Delta s^*\) governing routing preference:

\begin{theorem}[Existence and uniqueness of preference-preserving constraints for a pair of ingresses]
\label{theorem:3}
Given a client \(c_k\),
if it can enter either ingress $p_{i,j}$ or $p_{m,n}$ under different prepending lengths, 
there exists a unique integer \(\Delta s^*\), 
s.t.,\\
\( (p_{m,n} \succ_{c_k, (s_{m,n} - s_{i,j} = \Delta s^*)} p_{i,j}) \land (p_{i,j} \succ_{c_k, (s_{m,n} - s_{i,j} = \Delta s^* + 1)} p_{m,n}) \),
where $\succ_{c_k,\boldsymbol{S}}$ represents the partial routing preference of client \(c_k\) under prepending configuration $\boldsymbol{S}$, and $s_{i,j}$ represents the prepending length of $p_{i,j}$.
\end{theorem}

\begin{proof}
\textbf{(Existence).}
Since the client $c_k$ switches route (from $p_{m,n}$ to $p_{i,j}$) when the prepending length difference $s_{m,n} - s_{i,j}$ increases monotonically,
there must be an integer $d$ where the preference flips.

\medskip
\noindent
\textbf{(Uniqueness).}
Assume two distinct flip points $\Delta s_1 < \Delta s_2$ exist. 
This would imply that $c_k$ switches back to $p_{m,n}$ for some $\Delta s \in (\Delta s_1, \Delta s_2)$, violating BGP’s shortest-path selection principle~\cite{krenc2023coarse}, which shows that once a route becomes less preferred due to increased prepending, it cannot regain preference without external policy changes.
Thus, no such interval can exist, and $\Delta s^*$ must be unique.
Thus, there cannot be more than one such flip point, and $\Delta s^*$ must be unique.
\end{proof}

\subsection{Solving with constraints and contradiction resolution} \label{sec:tie-breaking}
In this section, we describe: 
(1) Given the constraints, how we solve the optimization to get an optimal prepending configuration; 
(2) How we efficiently resolve contradictions among constraints to ensure optimality. 
(3) An exceptional observation that ingress shifts can be caused by third-party prepending length. 


\noindent\textbf{Optimization solving implementation and feasibility verification.} \label{sec:formulation}
Given the desired client-ingress mapping $\boldsymbol{M^{*}}$,
we wish to determine an optimal prepending configurations $\boldsymbol{S}$ such that the product of $\boldsymbol{M^{*}}$ and $\boldsymbol{M}$ is maximized (that is to say, $\boldsymbol{M}$ matches $\boldsymbol{M^{*}}$ as much as possible), where $\boldsymbol{M}$ is the resulting client-ingress mapping under the prepending configurations $\boldsymbol{S}$. 
Accordingly, we formulate this problem as the program (1).

\begin{align*}
\text{maximize} & \sum_{c_k \in \boldsymbol{C}} \sum_{p_{i,j}\in \boldsymbol{P}} \boldsymbol{M^{*}}_{c_k,p_{i,j}} \cdot \boldsymbol{M}_{c_k,p_{i,j}}  \tag{1}\\
\text{s.t.} \quad & \sum_{p_{i,j}\in \boldsymbol{P}} \boldsymbol{M}_{c_k,p_{i,j}} = 1, \quad\forall c_k \in \boldsymbol{C} \tag{1a}\\
& \boldsymbol{M}_{c_k,p_{i,j}} \in \{0, 1\}, \quad\forall c_k \in \boldsymbol{C}, \forall p_{i,j}\in \boldsymbol{P}\tag{1b}\\
& \boldsymbol{M}_{c_k,p_{i,j}} \geq \boldsymbol{M}_{c_k,p_{m,n}}, \quad\forall p_{i,j} \succ_{c_k,\boldsymbol{S}} p_{m,n} \tag{1c}\\
& s_{i,j} \in \{0,1, \ldots, MAX\}, \quad\forall s_{i,j} \in \boldsymbol{S} \tag{1d}
\end{align*}

The constraint (1a) characterizes that one client enters exactly one ingress.
$M_{c_k,p_{i,j}}$ is a zero-one variable indicating whether the client $c_k$ enters the ingress $p_{i,j}$ or not.
The constraint (1c) preserves the partial order preference between $p_{i,j}$ and $p_{m,n}$ under the configuration $\boldsymbol{S}$.
The optimization variable $s_{i,j}$ is an integer, representing the prepending length configured for the ingress $p_{i,j}$.
This problem is NP-hard, which can be reduced from classic max-SAT problem~\cite{williamson2011design} and the proof can be found in Appendix~\ref{appendix:nphard}.

Before running program (1) using commercial solver OR-Tools~\cite{ortools}, we verify the solving feasibility with statistics.
First, although the total number of clients is large, most clients exhibit identical ingress selection patterns across configurations, enabling aggregation into \textit{client group} that the same set of routing constraints.
This grouping is derived empirically from observed routing behavior rather than predefined structures such as BGP atoms.
Specifically, despite there are \textasciitilde 2.4M clients, they form only \textasciitilde 14,700 unique client groups.
Second, sparse candidate ingress distribution is observed: 58\% of client groups have only 1–2 candidate ingresses (equivalently 0–1 constraints), while merely 15\% have over 10 candidate ingresses. 
Third, each client group’s constraints are structured in conjunctive normal form (CNF), requiring simultaneous satisfaction of all constraints. For example, constraints for a client group $c_a$ selecting desired ingress might be: 
$(s_{i,j} \leq s_{m,n} - MAX) \wedge (s_{i,j} \leq s_{x,y} - MAX)$.
while group $c_b$ choosing the desired ingress could require $(s_{i,j} \leq s_{m,n} - MAX) \wedge (s_{i,j} \leq s_{v,w} - MAX)$. 
Since the MAX-SAT problem is given by a CNF formula, our original problem is reducible to maximizing satisfaction of distinct atomic constraints of all client groups: $(s_{i,j} \leq s_{m,n} - MAX) \wedge (s_{i,j} \leq s_{x,y} - MAX) \wedge (s_{i,j} \leq s_{v,w} - MAX)$.
Therefore, the complexity of the problem is further reduced.

As a result, the total number of constraints is less than 1,000, and the solver can solve it in less than one second. 
For further scaling concerns, if the number increases significantly, we can prioritize large client groups to reduce the number or use the early stopping strategy to limit the inference time.

\vspace{0.1in}
\noindent\textbf{Constraint contradiction resolution.}
Preliminary constraints derived from \scan are categorized into two types based on their operational origins and structural forms: 
\textbf{TYPE-I constraints}, expressed as $s_{i,j} \leq s_{m,n} - MAX$, originate from scenarios where a client’s desired ingress $p_{i,j}$ becomes accessible only when its prepending length $s_{i,j}$ reaches zero.
\textbf{TYPE-II constraints}, expressed as $s_{i,j} \leq s_{m,n}$, arise when a client can access $p_{i,j}$ when both $p_{i,j}$ and $p_{m,n}$ are prepended $MAX$ ASes but turns to $p_{m,n}$ once $s_{m,n}$ depletes to zero. 

Consider two clients $c_k$ and $c_l$ whose preliminary constraints are $s_{i,j} \leq s_{m,n} - MAX$ and $s_{m,n} \leq s_{i,j}$ respectively.
As these two preliminary constraints cannot be satisfied simultaneously (the only solution for the first constraint is $s_{m,n}=0, s_{i,j}=MAX$ which is inherently unsatisfiable for the second constraint), we denote it as a constraint contradiction. 

Constraint contradictions arise from the preliminary constraints’ maximal looseness. By introducing tighter bounds ($\Delta s$), we resolve the contradictions through interval overlap analysis.
Taking the same example, for $c_k$'s constraints, according to Theorem~\ref{theorem:3}, there exists a $\Delta s_1^*$ which satisfies: $\Delta s_1^* \leq MAX$ and $s_{i,j} \leq s_{m,n} - \Delta s_1^*$. 
Therefore, the feasible interval for $c_k$ is $s_{m,n} - s_{i,j} \in [\Delta s_1^*, MAX]$. Similarly, there exists a $\Delta s_2^*$ which satisfies: $\Delta s_2^* \geq 0$ and $s_{m,n} \leq s_{i,j} + \Delta s_2^*$, with its feasible interval defined as $s_{m,n} - s_{i,j} \in [0, \Delta s_2^*]$. 
Consequently, the resolvability of the contradiction depends on the intersection of these intervals: $s_{m,n} - s_{i,j} \in [\Delta s_1^*, \Delta s_2^*]$. This intersection is non-empty if and only if $\Delta s_2^* \leq \Delta s_1^*$. If this condition holds, the constraints are resolvable; otherwise, they can be asserted unresolvable.

TYPE-II constraints are inherently resolvable between themselves because their mutual contradiction (\eg, $s_{i,j} \leq s_{m,n}$ and $s_{m,n} \leq s_{i,j}$) collapses into an equality $s_{i,j} = s_{m,n}$, which is always satisfiable.
Conversely, conflicting TYPE-I constraints are irreconcilable when mutually imposed(\eg, $s_{i,j} \leq s_{m,n} - MAX$ and $s_{m,n} \leq s_{i,j} - MAX$), as their combined logic enforces $MAX=0$, violating the predefined $MAX>0$ parameter. 
Based on this insight, we design a binary scan algorithm specifically targeting hybrid contradictions between TYPE-I and TYPE-II constraints.

\vspace{0.1in}
\begin{algorithm}
\caption{Binary scan for Constraint Resolution}
\begin{algorithmic}[1]
\algtext*{EndFor}
\algtext*{EndWhile}
\algtext*{EndIf}
\Statex \hspace{-\algorithmicindent}\textbf{Input:} 
Contradictions $\gamma_1: s_{i,j} \leq s_{m,n} - k $ and $\gamma_2: s_{m,n} \leq s_{i,j} + b$
\Statex \hspace{-\algorithmicindent}\textbf{Output:} Refined constraints $\gamma'_1$, $\gamma'_2$.

\State Initialize $\Delta s_1^{min} \gets 0$, $\Delta s_1^{max} \gets k$, $\Delta s_2^{min} \gets b$, $\Delta s_2^{max} \gets MAX$
\While{$\Delta s_1^{max} > \Delta s_2^{min}$ and $\Delta s_1^{min} < \Delta s_2^{max}$}
\If{$s_{i,j} \leq s_{m,n} - \frac{\Delta s_1^{min} + \Delta s_1^{max}}{2}$ still holds}
\State $\Delta s_1^{max} \gets \frac{\Delta s_1^{min} + \Delta s_1^{max}}{2}$
\Else
\State $\Delta s_1^{min} \gets \frac{\Delta s_1^{min} + \Delta s_1^{max}}{2}$
\EndIf
\If{$s_{m,n} \leq s_{i,j} + \frac{\Delta s_2^{min} + \Delta s_2^{max}}{2}$ still holds}
\State $\Delta s_2^{min} \gets \frac{\Delta s_2^{min} + \Delta s_2^{max}}{2}$
\Else
\State $\Delta s_2^{max} \gets \frac{\Delta s_2^{min} + \Delta s_2^{max}}{2}$
\EndIf
\EndWhile
\Return $s_{i,j} \leq s_{m,n} - s_1^{max}$, $s_{m,n} \leq s_{i,j} + s_2^{min}$
\end{algorithmic} \label{algo:ContraResol}
\end{algorithm} 
\vspace{0.1in}

\begin{figure}[tb]
    \includegraphics[width=.98\linewidth]{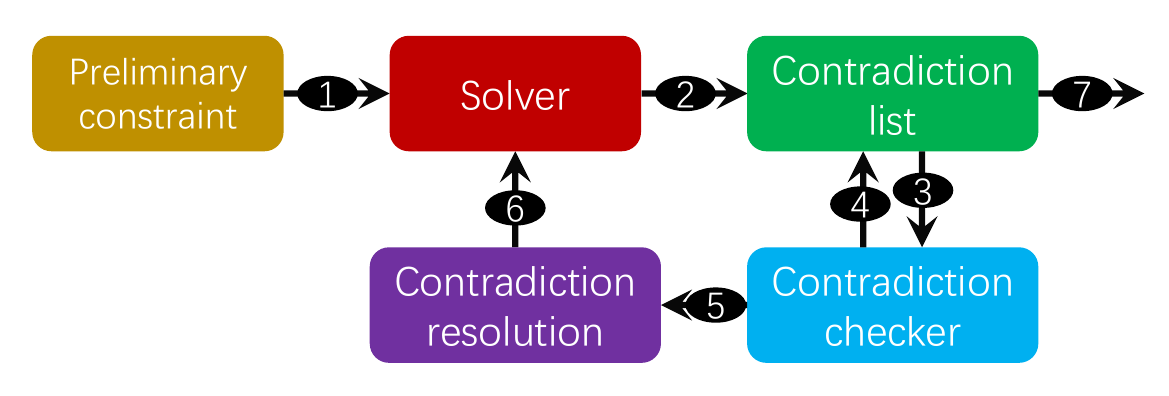}
    \vspace{4pt}
    \caption{Workflow of the contradiction resolution.
      \label{fig:resolution}
    }
\end{figure}

Algorithm~\ref{algo:ContraResol} resolves constraint contradictions through coordinated bisection refinement of slack variables ($\Delta s1$ and $\Delta s2$) for conflicting TYPE-I $\gamma_1$ and TYPE-II $\gamma_2$ constraints.
Initialized with $\Delta s1 \in [0, k]$ (experience prior binary scan adjustments but remain insufficiently tightened) and $\Delta s2 \in [b, MAX]$, the algorithm iteratively tightens both constraints by bisecting their respective threshold ranges, followed by client-ingress mapping validation to verify whether the modified constraints preserve desired client access patterns (lines 3\&7).
If validation fails, the search continues in subranges favoring larger $\Delta s1$ and smaller $\Delta s2$ to resolve the contradictions (lines 6\&10), otherwise, it narrows toward stricter adjustments (lines 4\&8).
The process ends either when overlapping intervals $\Delta s_1^{max} < \Delta s_2^{min}$ indicate resolvable conditions or when disjoint intervals $\Delta s_1^{min} > \Delta s_2^{max}$ confirm irreconcilable contradictions.
Each time, the adaptive binary scan strategically avoids the exact determination of \(\Delta s^*\).
Instead, it employs progressive boundary refinement by continuously narrowing feasibility regions until either the contradictions converge or are demonstrated to be unresolvable.
This dual bisection mechanism ensures minimal constraint relaxation while systematically probing the solution space.

\begin{figure}[tb]
    \includegraphics[width=.96\linewidth]{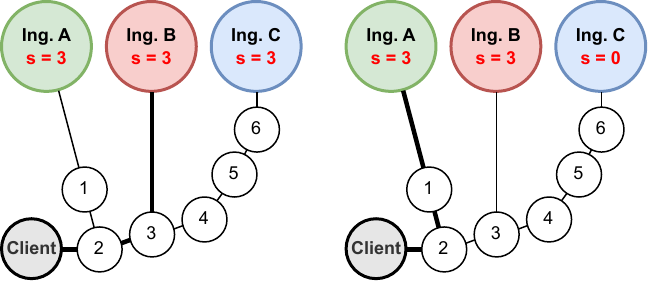}
    \vspace{4pt}
    \caption{Illustration of an ingress shift caused by the ASPP change of a third-party ingress .
      \label{fig:hiddenhand}
    }
\end{figure}

In Figure~\ref{fig:resolution}, we present a systematic constraint solution workflow that begins with the preliminary constraints generated from max-min polling. 
We first use a commercial solver OR-Tools to try to solve the constraint lists (\ding{182}).
If the constraint set is unsolvable, the solver identifies contradictory constraint pairs $\xi_i = (\gamma_{i1}, \gamma_{i2})$ that cannot coexist, forming a contradiction set $\boldsymbol{\Xi} = \{\xi_1, \xi_2, \dots\}$ (\ding{183}). These contradictions are prioritized by their client impact count, then evaluated through a verification process: first assessing whether either constraint in the pair constitutes a tight inequation (precisely bounded by $\Delta s$)(\ding{184}), where any tight constraint pairs automatically designates the contradiction as unresolvable (\ding{185}),  and second attempting resolution via binary scan optimization for untightened pairs (\ding{186}). 
After each resolution attempt, the constraint set is revalidated through solver re-execution (\ding{187}) to propagate resolution effects across interdependent constraints.
Since no additional contradictions would be generated, the entire set $\Xi$ can be processed in one pass.
The workflow terminates when all contradictions are either resolved or definitively classified as unresolvable, yielding an optimized prepending configuration through this closed-loop verification-optimization cycle (\ding{188}). Constraints processed through the binary scan are formally defined as finalized preference-preserving constraints.

\subsection{Deep Dive: Algorithm Design Rationale}
This section discusses the design rationale behind the algorithms used in \sys's operational deployment.

\noindent\textbf{Third-party impact:}When analyzing the results in \scan, we have identified two distinct types of clients that show different reactions to prepending adjustments:
(1) 95.1\% of client groups shift to one ingress when the prepending length of this ingress is tuned to smaller values, which aligns with intuitions and the Bi-ingress competing model in Theorem~\ref{theorem:3}.
(2) Conversely, 4.9\% clients make the ingress shift from B to A when the prepending length of an unrelated third-party C is adjusted. For instance, a client group at $\langle \text{California, AS132203} \rangle$ switches from the \pt in Frankfurt to Ashburn when the prepending length of Malaysia is changed to zero. 

To understand the causes behind the second category's behavior, we delve into an analysis of traceroutes from clients to their respective ingresses under different configurations. We demonstrate this phenomenon with an example in Figure \ref{fig:hiddenhand}. 
In this figure, each circle symbolizes an individual AS, and connecting lines between ASes indicate the propagation of announcements.
In this network with three ingresses A, B, and C. The client whose desired mapping is A receives anycast announcements. 
Suppose we are conducting \scan with $MAX=3$.
In the base mapping (left), ``AS 2'' prefers ingress B when the path to B is the shortest. 
When C is set as 0 (right), paths to C and A are both the shortest, but ``AS 2'' chooses A instead of C, which actually adjusts itself.
This behavior occurs because when AS path lengths are equivalent, ``AS 2'' has a bias for routes via ``AS 1'' rather than those via ``AS 3'' due to lower-tier-breaking metrics (e.g., origin code, MED or router ID)~\cite{cisco2024bgpattributes} of BGP.
Finally, ``AS 2'' will advertise the path originating from ingress A to the client.

\sys is compatible with this case through a new preference-preserving constraint format:

$(p_{m,n} \succ_{c, (s_{m,n} - s_{i,j} \leq \Delta s^*)} p_{i,j}) \land (p_{i,j} \succ_{c, (s_{m,n} - s_{i,j} > \Delta s^*)} p_{m,n})$, 

\noindent where $s_1$ and $s_2$ are no longer required to be $s_{i,j}$ and $s_{m,n}$, exactly matching $p_{i,j}$ and $p_{m,n}$. 
Instead they can be the prepending lengths of any other ingresses.
Note that our optimization formulation intrinsically allows this, i.e., the $\boldsymbol{S}$ in $M_{c_k,p_{i,j}} \geq M_{c_k,p_{m,n}}, \forall P_{i,j} \succ_{c_k,\boldsymbol{S}} P_{m,n}$ can be any prepending length array.
And the input to the solver is still in the same format.

\noindent\textbf{Middle ISP's impact:}
In empirical deployments of \sys, a primary operational concern arises from ISPs dynamically adjusting ASPP configurations --- either by truncating prepending length or reinflating AS paths. 
Actually this type of change on the ASPP configuration of the internal ISPs will not influence the correctness of our preference-preserving constraints,  provided ISP reactions to these configurations remain consistent --- a robustness validated by prior studies\cite{marcos2020path}.
This resilience stems from the system's inherent integration of client response modeling to ASPP changes, which inherently accounts for ISP-driven configuration fluctuations. 
For instance, if a middle ISP shortens the ASPP length to 3, our dynamic parameter tuning from 9 to 3 (if we set $MAX$ as 9) is transparent to clients 
--- no ingress shifts occur.
However, if further tuning ASPP to 2 triggers ingress changes, the system infers a preference-preserving constraint with $\Delta s=7$, which inherently accommodates both operator-initiated tuning and ISP-driven ASPP variations, ensuring stable routing behaviours despite configuration dynamics.

\section{Evaluation}

In this section, we first derive the computational complexity of \sys, and then validate it on the real-world anycast testbed described in \S~\ref{sec:mearsys}. 
Specifically, we aim to answer the following research questions.

\begin{itemize}[itemsep=1pt, topsep=1pt]
\item[RQ1]~How effective is \sys in reducing latency? How much gain is introduced by correctly inferring preference constraints?
\item[RQ2]~Does the optimization objective function value align well with latency? Does \sys correctly infer preference constraints? 
\item[RQ3]~What are the computational and operational complexity of \sys? What are the actual overhead and flexibility in production networks?
\item[RQ4]~How does \sys's optimization efficacy scale with varying numbers of enabled PoPs?
\end{itemize}

\begin{figure*}[tb]
    \subfigure[The proportion of clients exhibiting divergent reactions to the ASPP.]{
        \includegraphics[width=0.32\linewidth]{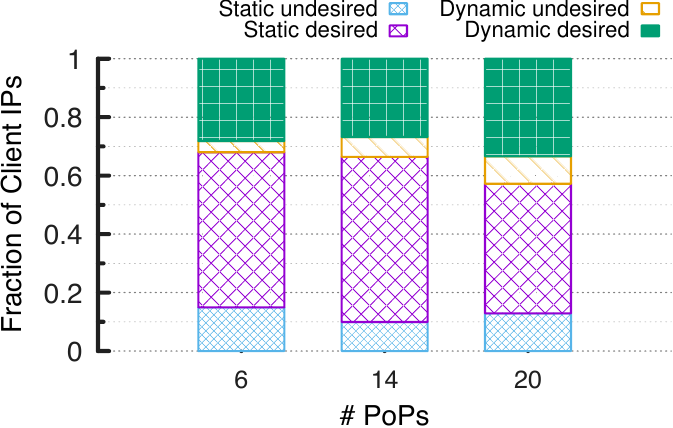}
        \label{fig:division}
    }
    \hspace{0.1cm}
    \subfigure[The distribution of client (groups) by their number of candidate ingresses.]{  
        \includegraphics[width=.32\linewidth]{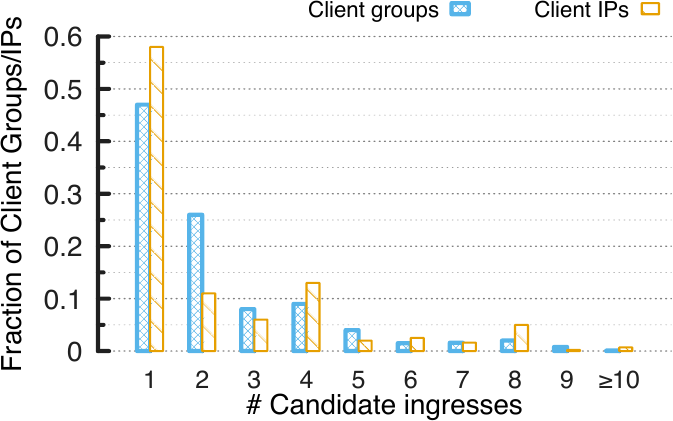}
        \label{fig:cop-route}
    }
    \hspace{0.1cm}
    \subfigure[The accuracy in predicting client accessibility to their desired PoPs.]{
        \includegraphics[width=.32\linewidth]{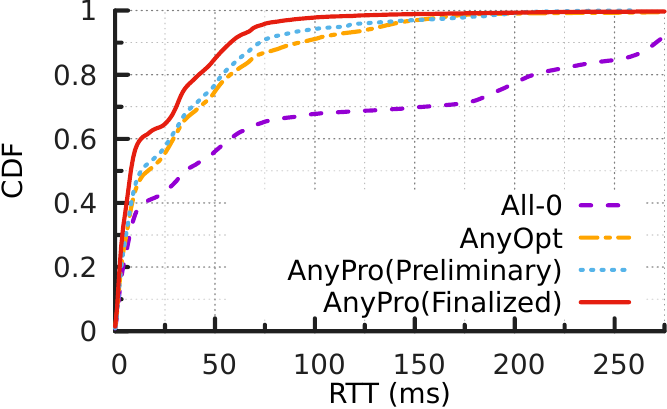}
        \label{fig:rtt}
    }
    \caption{\small (a) \sys can optimize performance for a substantial portion of clients. (b) Most clients and client groups are associated with only a single candidate ingress.
    (c) \sys-optimized configuration substantially outperforms other approaches in terms of RTT.
    }
    \vspace{-4pt}
\end{figure*}


\subsection{Overall Performance (RQ1)} \label{sec:performance}

\subsubsection{Experiment Setup.}
As outlined in \S~\ref{sec:model}, \sys\ requires desired client-to-ingress mappings as input, which are derived by operators from historical data and application-specific requirements; for our evaluation, we use geographical proximity as the primary mapping criterion to approximate latency, and we quantify anycast performance by comparing observed mappings against these geo-optimal ones.
\minyuan{To ensure the correctness of the observed client-ingress mappings, consecutive ASPP adjustments are spaced 10 minutes apart to allow the ASPP updates to fully propagate and stabilize across the global routing table before initiating ICMP probes. This duration is consistent with established inter-domain routing studies \cite{fastroute, zhang2021anyopt}, ensuring that we capture the steady-state routing behavior rather than transient path-hunting effects.
We specify $MAX = 9$ as our practical upper bound for prepending, a value informed by prior studies \cite{marcos2020path} and our empirical observations that transit providers commonly accept AS-path lengths up to this threshold without filtering.}

\noindent\textbf{Baselines.} We compare the following schemes with \sys.
\begin{itemize}[leftmargin=*, itemsep=0pt, parsep=0pt, topsep=0pt]

\item \textbf{All-0}: In this configuration, we enable all available ingress points without AS-path prepending (i.e., $\forall i,j$, $s_{i,j} = 0$).




\item \textbf{AnyOpt}: We additionally implement AnyOpt~\cite{zhang2021anyopt}, an anycast optimization approach that operates by selectively enabling/disabling PoP subsets. 

\item \textbf{AnyPro (Preliminary)}: 
We derive preference constraints from the output of \scan as described in §\ref{sec:model}, obtaining only preliminary constraints since we omit the resolution of the contradictions. From these constraints, we generate a preliminary ASPP configuration where each ingress's ASPP length is either 0 or 9.

\item \textbf{AnyPro (Finalized)}: 
After applying contradiction resolution using the method described in~\S\ref{sec:formulation}, we obtain the finalized preference constraints. In the resulting configuration, the ASPP of each ingress is selected from the discrete set \{0, 1, $\cdots$, 9 \}.

\end{itemize}

\noindent\textbf{Metrics.} We evaluate the optimization using both \textbf{RTT} and \textbf{\lar}, where
\textbf{\lar} is calculated as 
$\frac{\sum_{c_k \in \boldsymbol{C}} \sum_{p_{i,j}\in \boldsymbol{P}} \boldsymbol{M^{*}}_{c_k,p_{i,j}} \cdot \boldsymbol{M}_{c_k,p_{i,j}}}{|C|\cdot|P|}$. 
Since $|C|$ and $|P|$ are constants, this metric is essentially the optimization objective function itself, and we further demonstrate the correlation between these two metrics.
A normalized objective value closer to 1 indicates that the observed mappings are more closely aligned with the desired ones.

Our experiment first quantifies the client rerouting potential through ASPP by measuring how many clients can reach their desired ingresses. Figure~\ref{fig:division} presents the results of three anycast deployments. In the 20-PoP configuration, 57.2\% of clients maintain stable catchment PoPs during \scan, with 44.3\% reaching desired ingresses (static desired) and 12.9\% diverted to undesired ones (static undesired). Among the 42.8\% experiencing ingress shifts, 30.7\% ultimately connect to desired ingresses (dynamic desired) while 9.3\% do not (dynamic undesired). This yields a 77.8\% total \lar (static desired + dynamic desired).

\begin{table}[tb]
\caption{\Lar of the optimized anycast system across different methods, both with peers (w/ peer) and without peers (w/o peer).\label{tab:larscore}}
\centering
\begin{tabular}{ c c c c}
\multirow{2}{*}{\bf Method} & \multicolumn{2}{c}{\bf \LAR} \\\cmidrule(lr){2-3} & {w/o peer} & {w/ peer} \\
\midrule
All-0 & 0.60 & 0.68\\
AnyOpt & 0.66 & 0.76\\
AnyPro (Preliminary) & 0.72 & 0.82\\
AnyPro (Finalized) & 0.76 & 0.85 \\

\end{tabular}
\vspace{4pt}
\end{table}

Table \ref{tab:larscore} presents the \lar for \sys-optimized configurations. The left column shows \lar values excluding clients connected via peering links. For the All-0 deployment, the peer-inclusive configuration demonstrates a 0.08 lower \lar than the transit-only version. Since both configurations consider desired ingresses, this result suggests that substantial intercontinental routing occurs when using transit-only paths. Notably, \sys achieves competitive performance even without finalized preference constraints --- its \lar is merely 0.05 below ``AnyPro (Finalized)''. This effectiveness stems from two key factors: (1) Figure~\ref{fig:cop-route} reveals that most clients face $\leq$2 candidate ingresses, meaning few constraints are needed to ensure desired ingress connectivity, and (2) \sys's optimization strategically prioritizes high-weight constraints during contradiction resolution, preferentially serving the majority client base to maximize overall \lar.



\begin{figure}[t]
        \includegraphics[width=1\linewidth]{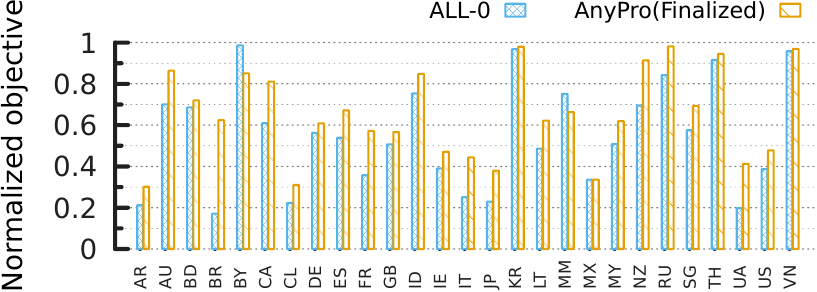}
        \vspace{1pt}
    \caption{\sys-optimized configuration can achieve high \lar for most countries simultaneously.\label{fig:breakdown}}
\end{figure}





\begin{figure*}[tb]
    \begin{minipage}[t]{0.65\textwidth}
    \vspace{-12pt}
    \subfigure[Matching accuracy vs. mean RTT.]{
        \includegraphics[width=0.5\linewidth]{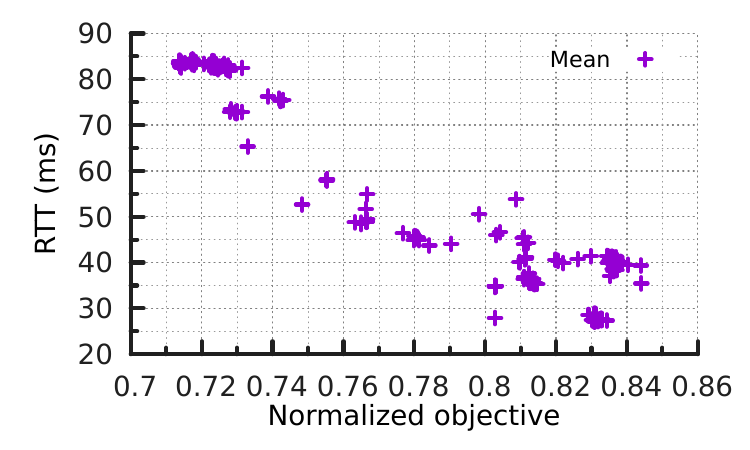}
        \label{fig:mean_rtt}
    }
    \subfigure[Matching accuracy vs. 95th percentile RTT.]{
        \includegraphics[width=0.5\linewidth]{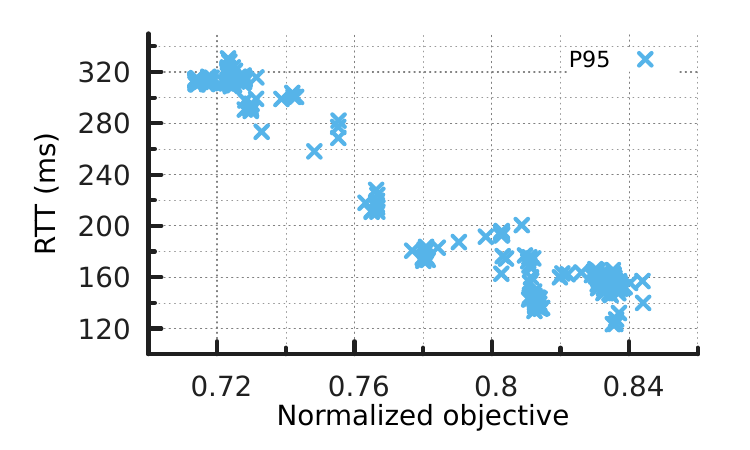}
        \label{fig:p95}
    }
    \caption{
    \minyuan{Correlation Analysis between \sys’s optimization objective (matching accuracy) and RTT performance.}\label{fig:relationship}}
    \end{minipage}
    \hfill
    \begin{minipage}[t]{0.32\textwidth}
    \subfigure{
        \includegraphics[width=0.95\linewidth]{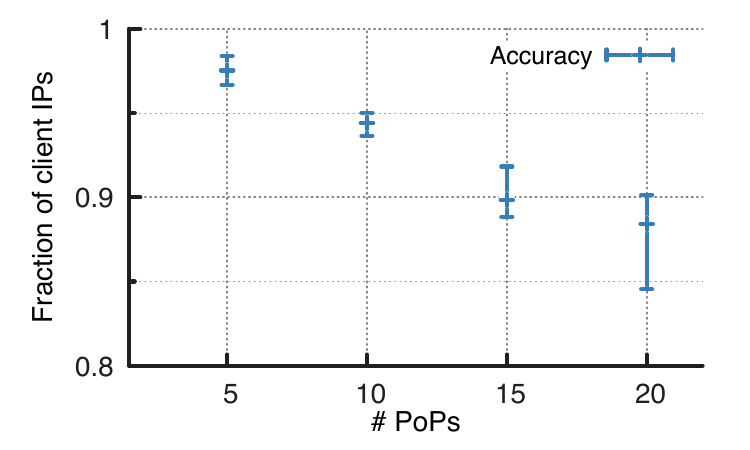}
    }
    \vspace{-0.2cm}
    \caption{The preference-preserving constraints derived by \sys are both complete and accurate.}\label{fig:errorbar}
    \end{minipage}
\end{figure*}

We further assess the RTT of \sys-optimized configurations against alternative deployments. Figure~\ref{fig:rtt} presents the CDF of client RTTs across different setups, showing that ``AnyPro (Finalized)'' achieves significantly reduced tail latency compared to All-0. The 90th percentile RTT improves from 271.2 ms (All-0) 
to 58.0 ms (``AnyPro (Finalized)''). 
This optimal performance stems from the two-stage optimization: AnyOpt first selects an optimal PoP subset, eliminating poorly-performing nodes, and \sys then fine-tunes ASPP values within this subset to precisely steer clients to the lowest-latency ingresses. These dramatic reductions demonstrate the ability of \sys, especially when combined with AnyOpt, to effectively leverage anycast routing for minimal latency.

To better quantify the effectiveness of \sys, we analyze country-level \lar improvements across \sys and AnyOpt deployments, focusing on the 27 countries with the highest transit-connected client populations. Figure \ref{fig:breakdown} reveals that \sys optimization enhances \lar for most countries, with Brazil showing the most dramatic improvement --- increasing from 0.17 to 0.62 under the ``AnyPro (Finalized)'' configuration. Our trace analysis indicates this stems from shifting Brazilian clients from Bangkok to Ashburn PoPs. The sole exception is Myanmar (MM bar), where \lar decreases as its lower traffic volume leads to deprioritization during constraint resolution, redirecting more than 10\% of Ho Chi Minh clients to European PoPs.

While effective, \sys faces limitations when handling contradictory constraints during optimization. The system prioritizes higher-weight client groups at the expense of smaller populations, potentially routing them to suboptimal PoPs. For instance, two competing constraints emerge: (1) 1,388 U.S. clients require $s_{({\text{Frankfurt,Telia}})} \geq s_{({\text{Ashburn,Level3}})} + 9 \geq 9$ to reach Ashburn-Level3, while (2) 467 German clients need $s_{({\text{India,  Airtel}})} \geq s_{({\text{Frankfurt,Telia}})} + 9$.
Given $s_{({\text{India, Airtel}})} \in [0,9]$, the system sets $s_{({\text{Frankfurt, Telia}})} = 0$, forcing the German clients to connect to undesired PoPs. This demonstrates how weight-based prioritization can disadvantage smaller client groups despite their legitimate routing preferences.


\noindent\textbf{Takeaways.} Our results demonstrate that ASPP-based optimization of IP anycast is both practical and effective. The \sys-optimized configuration achieves significant improvements, successfully redirecting the majority of clients to their desired PoPs. Furthermore, performance gains can be substantially enhanced through integration with AnyOpt, combining PoP selection optimization with intelligent path preference configuration to maximize anycast routing efficiency.


\subsection{Optimization Effectiveness  (RQ2)}

\subsubsection{Correlation between \lar and RTT}

To validate the alignment between \sys's optimization objective --- \lar --- and actual latency performance, we systematically analyze the relationship between \lar and RTT. \minyuan{Note that this analysis is specific to \sys's internal configuration space and is intended to verify that maximizing matching accuracy effectively minimizes RTT. }
As illustrated in Figures \ref{fig:mean_rtt} and \ref{fig:p95}, improvements in \lar consistently correspond to reductions in both mean RTT and the 95th-percentile RTT. 
This consistent inverse relationship is further substantiated through statistical analysis, revealing a strong negative correlation between \lar and RTT metrics. 
The Pearson correlation coefficients are approximately $-$0.95 for mean RTT and $-$0.96 for the 95th-percentile RTT\footnote{The Pearson correlation coefficient ranges from -1 to 1, where values approaching -1 indicate a strong inverse linear relationship.}, underscoring the reliability of this trend. 
These findings confirm that \lar serves as a robust and effective predictor of latency performance in anycast networks, capable of accurately reflecting both average and tail-end latency behavior. 
Higher \lar values are consistently associated with enhanced end-to-end latency characteristics across the entire performance distribution, thereby reinforcing its utility as a key metric for network optimization.



\noindent\textbf{Takeaways:} The \lar,  which serves as the optimization objective, exhibits a strong negative correlation with RTT.

\subsubsection{Constraint identification}

We evaluate \sys's effectiveness in identifying sufficient preference-preserving constraints to address RQ2. 

While \sys does not predict exact catchments for all possible configurations --- as exhaustively determining ingress preference orders would be computationally prohibitive --- it efficiently identifies preference-preserving constraints that guarantee client groups reach their desired PoPs by: (1) selecting a random subset of PoPs and disabling others to isolate ASPP-sensitive clients via \scan; (2) deriving constraints using our tie-breaking method (\S\ref{sec:tie-breaking}); and (3) validating accuracy through testing of 10 random ASPP configurations per deployment and comparing predicted versus observed PoP connections.

Figure~\ref{fig:errorbar} demonstrates the prediction accuracy variability across different deployment scales, where we evaluate \sys with 5, 10, 15, and 20 enabled PoPs (including all associated transit providers). The results show that \sys maintains strong predictive performance, particularly in smaller deployments --- achieving over 95\% accuracy with 5 enabled PoPs. However, as the number of PoPs increases, we observe a gradual decline in prediction accuracy for client-to-PoP mappings. This degradation primarily stems from two factors: (1) the quadratic growth in unresolved constraint contradictions as more PoPs are added (despite our binary scan-based resolution mechanism), and (2) external complexities such as uncommon BGP policies~\cite{zhang2021anyopt} and multi-path routing effects. Notably, even at 20 PoPs (our largest test case), \sys maintains 88.5\% accuracy, demonstrating its robustness for real-world anycast networks.



\noindent\textbf{Takeaways.} The preference-preserving constraints discovered by \sys are effective for determining client accessibility to their desired PoP.







\subsection{Complexity Analysis (RQ3)}
The initial \scan pass over the $n$ ingresses incurs a cost of only $\mathcal{O}(n)$. For $|\Xi|$ contradictions, the resolution process performs at most $|\Xi|$ binary scans, each requiring $\mathcal{O}(\log m)$ time, resulting in a total complexity of $\mathcal{O}(|\Xi| \log m)$. Overall, \sys completes in $\mathcal{O}(n + |\Xi| \log m)$ time, which is optimal and several orders of magnitude more efficient than the naïve $\mathcal{O}(m^n)$ approach. In our large-scale deployment scenario --- where $m = 10$, $n = 38$, and approximately 2.4 million IP addresses are considered --- the system identifies 513 distinct preliminary constraints and resolves all contradictions with only 84 ASPP adjustments. Combined with the 76 ASPP adjustments from the initial \scan process (i.e., $38 \times 2$), a full optimization cycle requires a total of 160 ASPP adjustments. \minyuan{This low-frequency, targeted approach ensures that \sys scales to large WANs without compromising Internet control-plane stability.}
With each adjustment taking 10 minutes to ensure convergence, the total computational time per cycle amounts to 26.6 hours --- a significant reduction from the 190 hours required by AnyOPT~\cite{zhang2021anyopt}. 
\minyuan{To validate the robustness of our derived constraints over the 26.6-hour window, 
we randomly sampled 50 non-contradiction preference constraints generated during the max-min polling, and derived a ASPP configurations to satisfy them simultaneously. 48 hours later, we reapplied the configurations.
By measuring the resulting client-ingress mappings again, we observed that 99.2\% of the mappings remained identical to the initial results and the 50 constraints still hold. 
This high degree of consistency indicates that inter-domain routing policies and path preferences exhibit significant persistence over multi-day periods, confirming the consistency of global constraints during each round, guaranteeing the effectiveness of \sys.}

\begin{figure}[t]
\centering
  \includegraphics[width=0.94\linewidth]{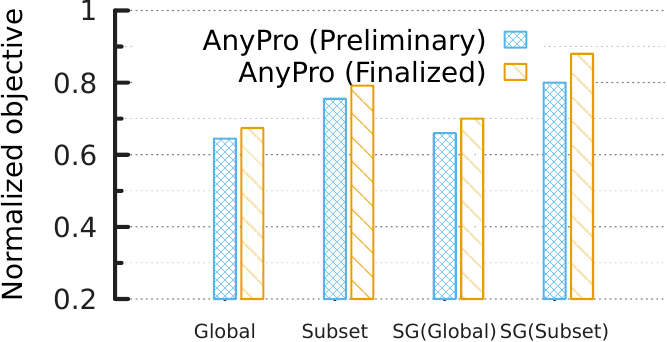}
    \caption{Subset optimization can yield greater performance for certain clients.}\label{fig:subset}
\end{figure}
\vspace{-0.1cm}

\subsection{Subset Optimization (RQ4)} 
 
Our prior analysis reveals that unresolved constraint contradictions disproportionately impact clients in low-traffic regions, reflecting a common operational practice that prioritizes service quality in specific critical areas over global optimization. To mitigate this imbalance, we propose to selectively deploy \sys across curated subsets of PoPs, allowing targeted optimization for priority client groups. We demonstrate that the optimization framework of \sys is particularly well-suited for several practical scenarios: (1) regionally constrained services, (2) regional IP anycast implementations~\cite{Zhou:sigcomm2023}, and (3) mitigation of temporary ingress outages. Using Southeast Asia as a case study, we activate six regional PoPs --- specifically in Malaysia, Manila, Ho Chi Minh City, Singapore, Indonesia, and Bangkok --- along with their transit links, while disabling all others. This configuration creates an isolated test environment in which we rigorously evaluate the efficacy of localized optimization.

Figure~\ref{fig:subset} compares \lar values for Southeast Asian countries under both global and localized optimization configurations. The leftmost bars present results from global optimization --- where all ingresses are enabled --- showing that these countries achieve a relatively lower \lar due to their reduced priority in worldwide routing policies. 
In contrast, localized subset optimization increases the overall \lar from 0.67 to 0.78, representing a 16.4\% improvement. Country-level analysis reveals that Singapore benefits most significantly, with its \lar rising from 0.70 to 0.88 --- a 25.7\% gain. 
A detailed analysis indicates that under global optimization, 16.6\% of Singapore clients are misrouted (12.5\% to U.S. and 4.1\% to Europe), whereas subset optimization successfully directs all traffic to local Singapore PoPs.

\newparagraph{Takeaways:} Localized optimization in Southeast Asia significantly improves routing accuracy, demonstrating \sys's capability for region-specific tuning. This is exemplified by Singapore’s 25.7\% \lar improvement after eliminating all transcontinental misroutes. The results confirm that \sys effectively enforces geographic prioritization while preserving the benefits of anycast.







\section{Discussion}\label{sec:discussion}

This section critically examines the inherent limitations of \sys and proposes promising avenues for future research to address these constraints.

\noindent\textbf{Why not Traceroute?} \label{sec:traceroute}
While traceroute-derived AS-path information could theoretically inform route optimization, a naive approach of comparing AS paths to get the finalized preference constraints faces two critical limitations. 
First, experimental data indicates that traceroute-collected path information often lacks completeness, particularly in documenting intermediate hops.
This makes it difficult to characterize the complete AS path for optimization purposes.
Second, certain ISPs implement BGP regular expression filters that dynamically truncate excessive route prepending — for instance, observed cases where 9$\times$ is compressed to 3$\times$Prepending - rendering direct AS-path length comparisons invalid, as the ISP’s real-time processing of ASPP remains opaque.
Our methodology overcomes these challenges by employing empirical performance evaluation (via max-min polling and binary scan) to derive both preliminary and finalized constraints. 
This approach ensures reliable client-side routing predictability independent of opaque ISP ASPP handling mechanisms.


\begin{figure}[tb]
    \includegraphics[width=0.98\linewidth]{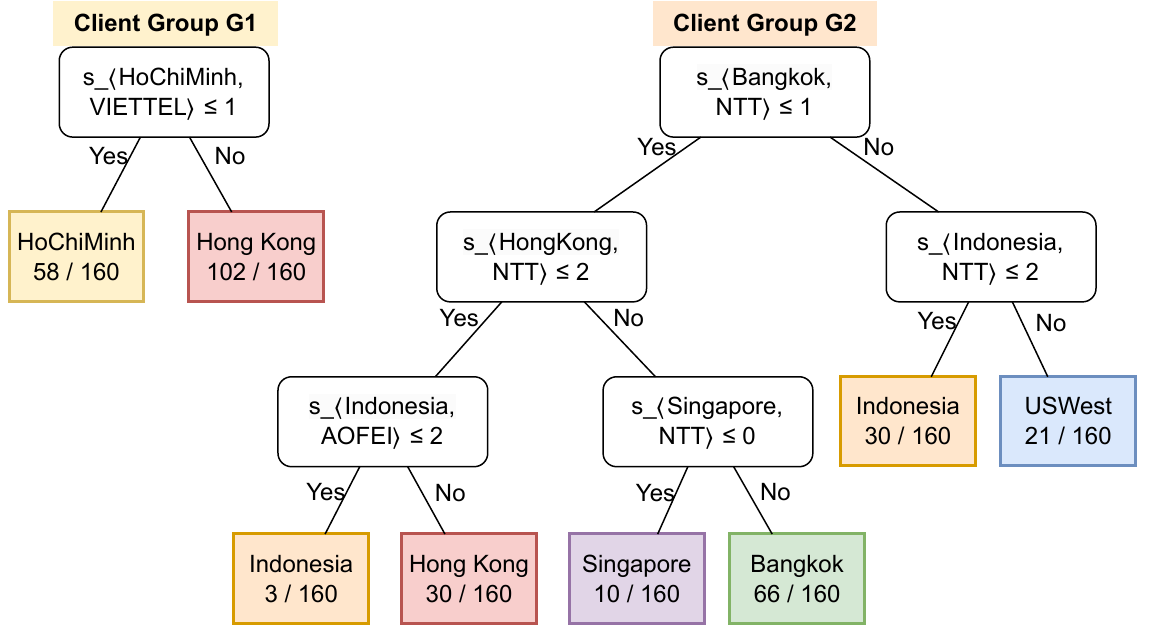}
    \vspace{1pt}
    \caption{Instability of decision tree models during catchment prediction.\label{fig:tree}}
    \vspace{2pt}
\end{figure}

\noindent\textbf{Data-driven catchment modeling and inference.}
\sys employs a deterministic approach for catchment modeling, inference, and optimization. While data-driven methods like machine learning (ML) may appear appealing, we demonstrate that ML catchment inference proves fundamentally unreliable and inefficient. As an illustrative case, we trained decision tree models using 160 random ASPP configurations and their corresponding client-ingress mappings as training data. The models take prepending length arrays as input and predict client-ingress relationships. Figure~\ref{fig:tree} visualizes the resulting tree structures for two representative client groups: G1 in Vietnam (2 candidate ingresses) and G2 in Indonesia (6 candidate ingresses), represented as leaves in their respective trees. This exercise reveals the inherent limitations of probabilistic approaches compared to our deterministic methodology.

The decision tree for G1 predicts with 100\% confidence that clients enter Ho Chi Minh City when $s_{({\text{HoChiMinh,VIETTEL}})}\leq$ 1, and Hong Kong otherwise. However, this proves unreliable when tested with \(s_{({\text{HoChiMinh,VIETTEL}})} = 1\) and $s_{({\text{HongKong, NTT}})}$ $=$ 0, where clients undesiredly connect to Hong Kong. These inconsistencies worsen for complex cases like G2 (with 6 candidate ingresses), and frequently occur when clients ingress through PoPs absent from training data. We identify two root causes: (1) BGP policies are fundamentally deterministic and ill-suited to probabilistic modeling, and (2) random ASPP configurations fail to capture sensitivity and constraint contexts. \sys overcomes these limitations by systematically discovering all possible ingresses and their constraints, guaranteeing optimal configurations. 

More advanced machine learning methods, such as graph neural networks (GNNs)~\cite{ferriol2023routenet} that model global correlations in evolutionary optimization~\cite{ouyang2024learn} and hybrid CNN–GRU–GA pipelines for structured feature learning and parameter search~\cite{ke2025stable}, demonstrate that graph-structured learning can effectively capture complex network dependencies.
While they may exceed the decision tree, and potentially approximate the optimal, they cannot match \sys's efficiency, as every configuration exploration in \sys provides essential information.


\noindent\textbf{Comparison with Alternative BGP Controls.}\label{sec:discuss community}
\minyuan{While BGP communities and announcement scoping (e.g.,  restricting route propagation to specific regions) offer alternative anycast steering, they present challenges in granularity and operational consistency. Unlike binary scoping, which is limited to discrete PoP-level advertisements, \sys leverages the scalar nature of ASPP (prepending $1 \dots 9$ times)) to enable \textit{fine-grained} steering of catchment boundaries without sacrificing global reachability. 
Furthermore, while we acknowledge that higher-priority policies like \texttt{Local-Pref} (often set via communities) can override AS-path length, \sys is specifically designed to empirically isolate \textit{ASPP-sensitive} clients through max-min polling (\S~\ref{sec:model}). Clients behind ISPs with rigid, community-driven policies that mask prepending effects are naturally identified as non-sensitive and excluded from the optimization variables. This focus on the universal AS-path attribute avoids the operational complexity of non-standardized communities across heterogeneous ISPs, though integrating community-based controls for anomaly detection remains a promising direction for future work.}

\noindent\textbf{Peering connections.}
Our anycast network leverages both transit and peering connections, which differ fundamentally in their routing characteristics and performance impacts. While prior work typically separates peer optimization from transit configuration (e.g., AnyOpt's sequential approach), we adopt a distinct strategy by enabling all peering connections before transit optimization. This design reflects two key considerations: (1) Performance benefits --- despite known risks of remote peering~\cite{Castro2014remotepeering,Nomikos2018Peer}, most peering paths in practice offer lower latency than transit alternatives, and (2) Relationship preservation --- frequent prefix announcement changes may violate peering agreements that stipulate route stability. Our approach thus maximizes potential performance gains while maintaining critical peering relationships.


\section{Related works}
\noindent{\bf Anycast catchment measurement and inference.} 
Verfploeter~\cite{Verfploeter2017de} supports inferring any clients' catchment sites as long as the clients are capable of responding to ICMP requests. Leveraging this capacity, researchers have extended its utility to assess client-specific metrics such as RTTs, traceroutes, and ingresses to anycast PoPs~\cite{de2020global,zhang2021anyopt,Many2020Sommese,alfroy2022mvp}. 
Another research approach utilizes RIPE Atlas to send active probes towards an anycast address~\cite{Li2018Internet,Wei2017does,de2017Anycast,calder2015analyzing}. For instance, Zhou \etal~\cite{Zhou:sigcomm2023} infer the anycast catchment by interpreting the penultimate hop from traceroute results, which typically corresponds to the on-site routers situated at the same locale as the anycast PoPs. 
Despite its wide utility in research, the RIPE Atlas method is constrained by the limited availability of probes.
A related research line focuses on inferring catchments for Internet routing~\cite{catchmentinfer-sermpezis-sigmetrics19,singh2021predict}.
Sermpezis \etal~\cite{catchmentinfer-sermpezis-sigmetrics19} suggest utilizing inferred AS-level Internet topology as a means of predicting catchments. However, this methodology does not effectively scale to the magnitudes of anycast networks and cannot guarantee accurate inferences at such scales. Singh \etal~\cite{singh2021predict} train a probabilistic model for Internet routing by feeding a corpus of traceroutes. While insightful, this approach also falls short when applied to a practical anycast networks.


\noindent{\bf Anycast catchment optimization.}
Existing catchment optimization can be broadly classified into two categories: PoP-level~\cite{zhang2021anyopt,Zhou:sigcomm2023,Li2018Internet,fastroute,Sarat2006On} and ingress-level~\cite{Stephen2019Taming,ballani2005towards,Alzoubi2011Association} optimization. (1) For the former, Zhang \etal~\cite{zhang2021anyopt} construct the total preference order of anycast sites for each client. This strategy helps forecast anycast catchments and effectively minimizes RTT, at the expense of time-consuming measurements. Zhou \etal~\cite{Zhou:sigcomm2023} deeply explore the potential of regional anycast, concluding that it can mitigate path inflation and enhance the performance over global anycast in practice.
(2) For the latter, DailyCatch~\cite{Stephen2019Taming} captures the performance shift across varying snapshots, each representing a particular configuration. 
Ballani \etal~\cite{ballani2005towards} suggest connecting all anycast PoPs to a single tier-1 provider, leveraging the intra-domain routing for optimal PoP selection. 
Bertholdo et al.~\cite{bertholdo2020bgp} aim to manage anycast catchment by adjusting BGP communities. However, their work primarily showcases the performance of specific configurations without deriving the optimal one. Their approach thus functions more as a monitoring tool rather than an optimization method.
To compare, \sys constructs a white-box model that encapsulates the critical conditions needed for each ASPP-sensitive client to connect to their expected ingress. 
\sys has been demonstrated to be effective across any subset of ingresses so that it can work well with ingress-level optimization.






\vspace{-2pt}
\section{Conclusion}

This paper introduces \sys, the first production-ready system that optimizes anycast catchment efficiency through strategic AS-path prepending configuration tuning. \sys efficiently characterizes global catchments and performs \scan to identify ASPP-sensitive clients and derives preliminary routing preference constraints. If there arises contradiction among constraints, it conducts iterative contradiction resolution and optimization solving to obtain the optimal global ASPP configuration. Evaluations on a production testbed confirm the effectiveness of \sys with significant latency improvements.

\vspace{-2pt}
\section{Acknowledgment}
Minyuan Zhou, Yuning Chen and Yifei Xu were supported by Alibaba Group through Alibaba Research Intern Program. We also sincerely thank our shepherd Oliver Hohlfeld and all anonymous reviewers for their valuable feedback.

\newpage
\bibliographystyle{plain}
\bibliography{references}

\appendix
\appendix
\section*{Appendix}

\section{Ethical Considerations.} 
We issue the ICMP requests at a reasonably low rate to avoid additional load on the Internet infrastructure. 
Additionally, our BGP announcements are restricted to prefixes used for testing and using our own AS number. 
In system logs, only IP, timestamps, and ingresses are parsed for analysis. User information is not included.
This work raises no other ethical issues.

\section{Infrastructure} \label{appendix:infra}

\begin{table}[t]
\begin{tabular}{l | l}
\hline
Malaysia   & NTT\_2914, AIMS\_24218                      \\ \hline
Madrid    & TATA\_6453                                     \\ \hline
Manila    & PLDT-iGate\_9299, Globe\_4775               \\ \hline
Hong Kong     & PCCW\_3491, NTT\_2914                       \\ \hline
Seoul     & SKB\_9318, TATA\_6453                       \\ \hline
Vancouver    & TATA\_6453                                      \\ \hline
Ashburn    & Level3\_3356, Cogent\_174                   \\ \hline
Moscow    & Rostelecom\_12389, Megafon\_31133           \\ \hline
Chicago    & CenturyLink\_3356, Cogent\_174              \\ \hline
Ho Chi Minh    & VIETTEL\_7552, CMC\_45903                   \\ \hline
California     & NTT\_2914, TATA\_6453                     \\ \hline
Frankfurt   & Telia\_1299, TATA\_6453                     \\ \hline
Bangkok     & TATA\_6453, TrueIntl.Gateway\_38082 \\ \hline
Singapore    & Singtel\_7473, TATA\_6453, PCCW\_3491    \\ \hline
Sydney     & Telstra\_4637, Optus\_7474                  \\ \hline
Toronto    & TATA\_6453                                     \\ \hline
India     & TATA\_4755, Airtel\_9498                    \\ \hline
Indonesia     & NTT\_2914, AOFEI\_135391                    \\ \hline
London     & TATA\_4755, Telia\_1299                   \\ \hline
Tokyo     & NTT\_2914, SoftBank\_17676                  \\ \hline
\end{tabular}

\caption{Transits and ASNs in the testbed for each PoP. Some are listed after a city while others are after a country.}
\label{table:testbed_asn}
\end{table}

Table~\ref{table:testbed_asn} presents the PoPs along with their respective transit providers used in our study.


\begin{figure}[tb]
    \includegraphics[width=.96\linewidth]{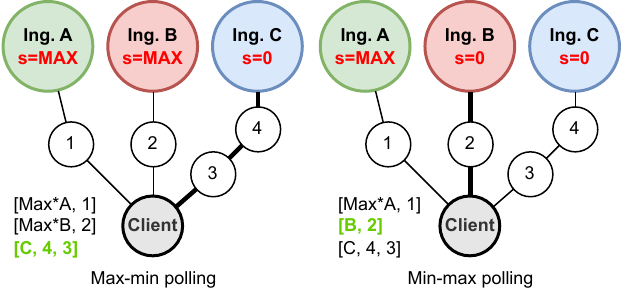}
    \vspace{4pt}
    \caption{Difference between \scan and \textit{min-max polling}. \textit{min-max polling} will not explore the path from C for the client because there will always be a shorter path available, either from A or B.
      \label{fig:90scan}
    }
\end{figure}
\section{Why not Min-Max Polling} \label{appendix:max-min}
The \scan offers a clear method to analyze the individual \pt's effect on traffic catchment areas. 
Figure~\ref{fig:90scan} illustrates the rationale behind choosing a \scan instead of a \textit{min-max polling} (\ie, initially prepending zero to all route announcements and then sequentially reverting to $\textit{Max}$ for all ingresss one by one). 
Specifically, consider a client that receives routing information from three different PoPs and AS path length is the most important route selection criterion. 
When prepending sizes are zero, routes from A and B always have shorter AS paths than that from C. Thus, in the case of \textit{min-max polling}, path originating from C would never be selected by the client because there will always be a preferable route available via either A or B. However, by employing \scan, we are able to explore the scenario where the client might select the path from C when routes from A and B both prepend 3 ASes. 

\section{Reduction from Max-SAT to ASPP Optimization}\label{appendix:nphard}

In this section, we show why and how we reduce NP-hard max-SAT problem into the ASPP optimization.

Taking another look into preference preserving constraints in the formal manner, we can express them as clauses of boolean values.
The preference preserving constraints can be formally expressed as clauses of boolean values.
The preference preserving of each client is a clause, consisting of booleans connected with \textit{conjunction} logic. 
Each boolean is a pairwise preference preserving constraint, which is an inequation involving two prepending sizes. 
We use conjuction logic because for clients with more than two candidate ingresses, to guarantee matching, several inequation constraints need to be satisfied simultaneously, so that the preferred ingress wins over every other ingresses.
To this end, it follows the nature of MAX-SAT problem, which we try to satisfy as much clauses as possible. 




\begin{proof}
Consider an arbitrary Max-SAT instance given by a CNF formula
\[
\Phi = \bigwedge_{k=1}^{m} C_k, \quad \text{with } C_k = \bigvee_{j=1}^{r_k} \ell_{k,j}.
\]
For each clause \(C_k\), create a client \(c_k\); for each literal \(\ell_{k,j}\) in \(C_k\), define a candidate preference \(p_{k,j}\) with weight
\[
\mathbf{M}^*_{c_k,p_{k,j}} = 1.
\]
Formulate the ASPP optimization instance as
\begin{align*}
\text{maximize} & \sum_{c_k \in \boldsymbol{C}} \sum_{p_{i,j}\in \boldsymbol{P}} \boldsymbol{M^{*}}_{c_k,p_{i,j}} \cdot X_{c_k,p_{i,j}}  \tag{1}\\
\text{s.t.} \quad & \sum_{p_{i,j}\in \boldsymbol{P}} X_{c_k,p_{i,j}} = 1, \quad\forall c_k \in \boldsymbol{C} \tag{1a}\\
& X_{c_k,p_{i,j}} \in \{0, 1\}, \quad\forall c_k \in \boldsymbol{C}, \forall p_{i,j}\in \boldsymbol{P}\tag{1b}\\
& X_{c_k,p_{i,j}} \geq X_{c_k,p_{m,n}}, \quad\forall p_{i,j} \succ_{c_k,\boldsymbol{S}} p_{m,n} \tag{1c}\\
& s_{i,j} \in \{0,1, \ldots, MAX\}, \quad\forall s_{i,j} \in \boldsymbol{S} \tag{1d}
\end{align*}

Considering each constraint satisfaction as a clause $C_k$, and reduce the exactly logical negations of each other as $\neg C_k$.
Then any solution to this ASPP instance corresponds to a truth assignment for \(\Phi\) that satisfies as many clauses as possible. Since the reduction is performed in polynomial time and Max-SAT is NP-hard, ASPP optimization is NP-hard.
\end{proof}
\end{document}